\documentclass[oribibl]{llncs}
\usepackage{amsmath, amssymb, algorithm, algorithmic, color}
\usepackage{enumerate, graphicx, hyperref,wrapfig}
\usepackage{cite}

\let\doendproof\endproof
\renewcommand\endproof{~\hfill\qed\doendproof}

\newcommand{\klevel}{\kappa}
\newcommand{\maxklevel}{\kappa_{\max}}

\title{Drawing Arrangement Graphs In Small Grids,\\
Or How To Play Planarity}
\author{David Eppstein\thanks{This research was supported in part by NSF grants
0830403 and 1217322 and by the Office of Naval Research under grant
N00014-08-1-1015.}}
\institute{Department of Computer Science, University of California, Irvine}

\begin{document}
\maketitle

\begin{abstract}
We describe a linear-time algorithm that finds a planar drawing of every graph of a simple line or pseudoline arrangement within a grid of area $O(n^{7/6})$. No known input causes our algorithm to use area $\Omega(n^{1+\epsilon})$ for any $\epsilon>0$; finding such an input  would represent significant progress on the famous $k$-set problem from discrete geometry. Drawing line arrangement graphs is the main task in the \emph{Planarity} puzzle.
\end{abstract}

\pagestyle{plain}

\section{Introduction}

\begin{wrapfigure}[14]{r}{0.45\textwidth}
\vskip-4.8ex
\centering\includegraphics[scale=0.225]{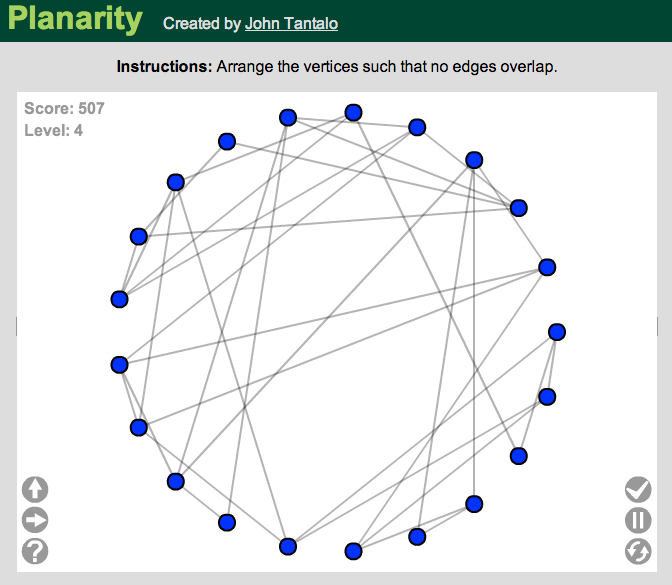}
\vskip-1.5ex
\caption{Initial state of Planarity.}
\label{fig:planarity-input}
\end{wrapfigure}

\emph{Planarity} (\url{http://planarity.net/}) is a puzzle  developed by John Tantalo and Mary Radcliffe in which the user moves the vertices of a planar graph, starting from a tangled circular layout (Figure~\ref{fig:planarity-input}), into a position where its edges (drawn as straight line segments) do not have any crossings. The game is played in a sequence of levels, of increasing difficulty. To construct the graph for the $i$th level of the game, the game applet chooses $\ell=i+3$ random lines in general position in the plane. It creates a vertex for each of the  $\ell(\ell-1)/2$ crossings of two lines in this arrangement, and an edge for each of the $\ell(\ell-2)$ pairs of  crossings that are consecutive on the same line. 

\begin{figure}[t]
\centering\includegraphics[scale=0.2]{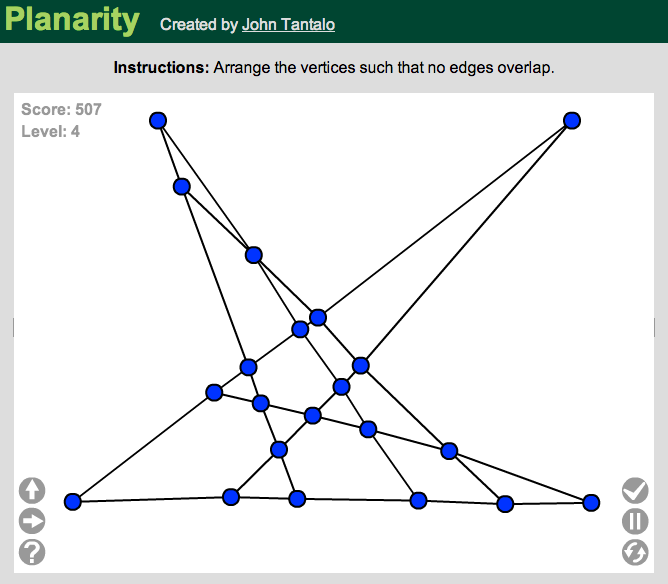}\qquad
\includegraphics[scale=0.2]{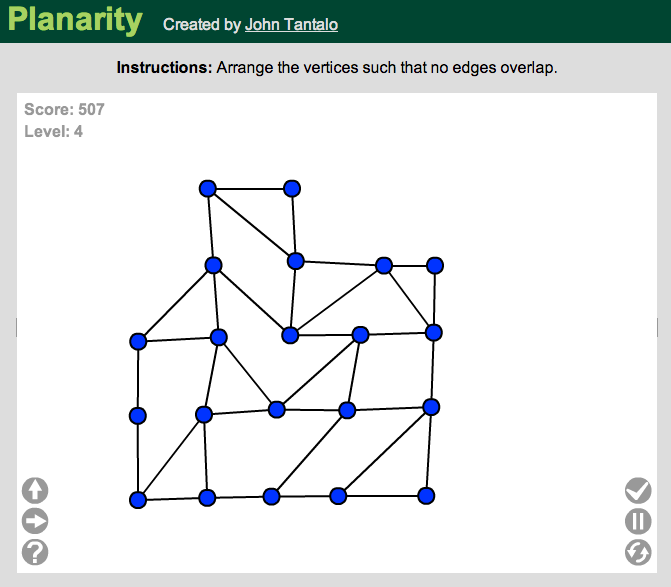}
\caption{Two manually constructed solutions to the puzzle from Figure~\ref{fig:planarity-input}. Left:  a set of lines with this graph as its arrangement. Right: an (approximate) grid layout.}
\label{fig:planarity-solved}
\vskip-2ex
\end{figure}

One strategy for solving Planarity would be to reconstruct a set of lines forming the given graph (Figure~\ref{fig:planarity-solved}, left). However, this is tedious to do by hand, and has high computational complexity: testing whether an arrangement of curves can be stretched to form a combinatorially equivalent line arrangement is NP-hard~\cite{Sho-VKF-91}, from which it follows that recognizing the graphs of line arrangements is also NP-hard~\cite{BosEveWis-IJCGA-03}. More precisely, both problems are complete for the existential theory of the reals~\cite{Sch-GD-09}. And although drawings constructed in this way accurately convey the underlying construction of the graph, they have low angular resolution (at most $\pi/\ell$) and close vertex spacing, making them hard to read and hard to place accurately by hand.
Instead, in practice these puzzles may be solved more easily by an incremental embedding strategy that maintains a planar embedding of a subgraph of the input, starting from a single short cycle (such as a triangle or quadrilateral), and that at each step extends the embedding by a single face, bounded by a short path connecting two  vertices on the boundary of the previous embedding. When using this strategy to solve a Planarity puzzle, the planar embedding may be kept tidy by placing each vertex into an approximate grid (Figure~\ref{fig:planarity-solved}, right). Curiously, the grid drawings found by this incremental grid-placement heuristic appear to have near-linear area; in contrast, there exist planar graphs such as the \emph{nested triangles graph} that cannot be drawn planarly in a grid of less than $\Theta(n^2)$ area~\cite{DolLeiTri-ACR-84,Val-TC-81}.

In this paper we explain this empirical finding of small grid area by developing an efficient algorithm for constructing compact grid drawings of the arrangement graphs arising in Planarity. Because recognizing line arrangement graphs is NP-hard, we identify a larger family of planar graphs (the graphs of simple pseudoline arrangements) that may be recognized and decomposed into their constituent pseudolines in linear time.  We show that every $n$-vertex simple pseudoline arrangement graph may be drawn in linear time in a grid of size $\maxklevel(O(\sqrt n))\times O(\sqrt n)$; here $\maxklevel(\ell)$ is
 the maximum complexity of a $k$-level of a pseudoline arrangement with $\ell$ pseudolines~\cite{KlaPatPip-82,TamTok-Algo-03,ShaSmo-WADS-03}, a topological variant of the famous $k$-set problem from discrete geometry (see Section~\ref{sec:grid} for a formal definition). The best proven upper bounds of $O(\ell\,^{4/3})$ on the complexity of $k$-levels~\cite{Dey-DCG-98,TamTok-Algo-03,ShaSmo-WADS-03} imply that the grid in which our algorithm draws these graphs has size $O(n^{2/3})\times O(\sqrt n)$ and area $O(n^{7/6})$. However, all known lower bounds on $k$-level complexity are of the form $O(\ell^{1+\epsilon})$ for all $\epsilon>0$~\cite{KlaPatPip-82,Tot-DCG-01}, suggesting that our algorithm is likely to perform even better in practice than our worst-case bound. If we could find a constant $\epsilon>0$ and a family of inputs that would cause our algorithm to use area $\Omega(\ell^{1+\epsilon})$, such a result would represent significant progress on the $k$-set problem.

We also investigate the problem of constructing \emph{universal point sets} for arrangement graphs, sets of points that can be used as the vertices for a straight-line planar drawing of every $n$-vertex arrangement graph. Our construction directly provides a universal point set consisting of $O(n^{7/6})$ grid points; we show how to sparsify this structure, leading to the construction of a universal set of $O(n\log n)$ points in the form of a subset of a grid whose dimensions are again $O(n^{2/3})\times O(\sqrt n)$.

Finally, we formalize and justify an algorithm for manual solution of these puzzles that greedily finds short cycles and adds them as faces to a partial planar embedding. Although this algorithm may fail for general planar graphs, we show that for arrangement graphs it always finds a planar embedding that is combinatorially equivalent to the original arrangement.

\section{Preliminaries}

\begin{wrapfigure}[16]{r}{0.43\textwidth}
\vskip-4.5ex
\centering\includegraphics[scale=0.2]{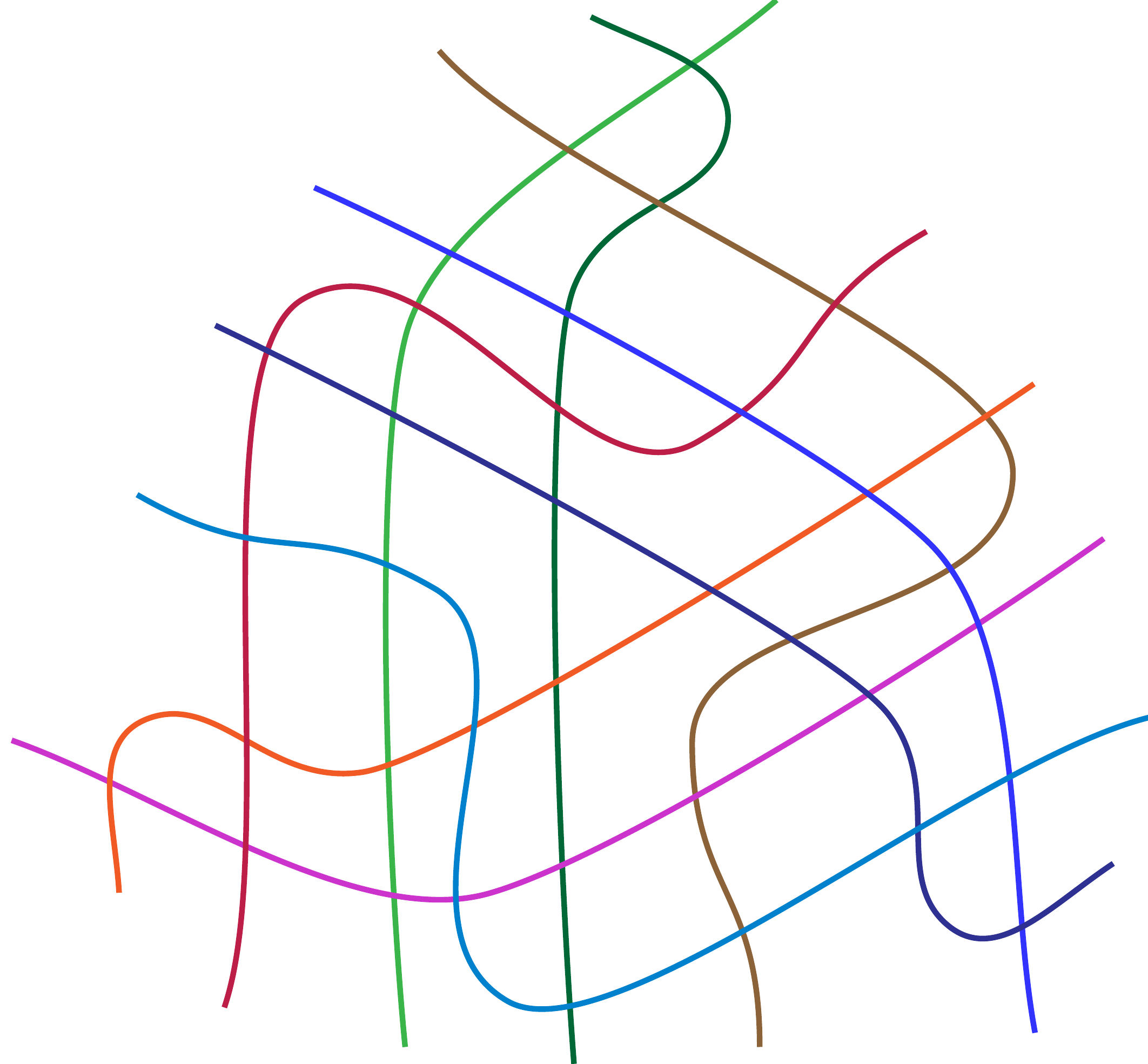}
\vskip-1.5ex
\caption{A simple pseudoline arrangement that cannot be transformed into a line arrangement. Redrawn from Figure 5.3.2 of \cite{Goo-HDCG-97}, who attribute this arrangement to Ringel.}
\label{fig:unstretchable}
\end{wrapfigure}

Following Shor~\cite{Sho-VKF-91}, we define a \emph{pseudoline} to be the image of a line under a homeomorphism of the Euclidean plane. Pseudolines include lines, non-self-crossing polygonal chains starting and ending in infinite rays, and the graphs of continuous real functions. A \emph{crossing} of two pseudolines is a point $x$ such that a neighborhood of $x$ can be mapped homeomorphically to a neighborhood of the crossing point of two lines, with the map taking the pseudolines to the lines within this neighborhood. An \emph{arrangement} of pseudolines is a finite set of pseudolines, the intersection of every two of which is a single crossing point. An arrangement is \emph{simple} if all pairs of pseudolines have distinct crossing points. A \emph{pseudoline arrangement graph} is a planar graph whose vertices correspond to the crossings in a simple pseudoline arrangement, and whose edges connect consecutive crossings on a common pseudoline.

Most of the ideas in the following result are from Bose et al.~\cite{BosEveWis-IJCGA-03},
but we elaborate on that paper to show that linear time recognition of arrangement graphs is possible. (See~\cite{Epp-GD-04} for a more complicated linear time algorithm that recognizes the dual graphs of a wider class of arrangement graphs, the graphs of \emph{weak} pseudoline arrangements in which pairs of pseudolines need not cross)

\begin{lemma}
\label{lem:test-arrangement}
If we are given as input a graph $G$, then in linear time we can determine whether it is a pseudoline arrangement graph, determine its (unique) embedding as an arrangement graph, and find a pseudoline arrangement for which it is the arrangement graph.
\end{lemma}

\begin{proof}
Let $G^*$ be formed from a pseudoline arrangement graph $G$ by adding a new vertex $v_\infty$ adjacent to all vertices in $G$ of degree less than four. As Bose et al.~\cite{BosEveWis-IJCGA-03} show, $G^*$ is 3-connected and planar, and its unique planar embedding is compatible with the embedding of $G$ as an arrangement graph. For convenience we include two edges in $G^*$ from $v_\infty$ to each degree two vertex in $G$, so that, in $G^*$, all vertices except $v_\infty$ have degree four. With this modification, the pseudolines of the arrangement for $G$ are represented in $G^*$ by paths starting and ending at $v_\infty$ that, at each other vertex, connect two opposite edges in the embedding.

For any given graph $G$ of maximum degree four we may, in linear time, add a new vertex $v_\infty$, test planarity of the augmented graph $G^*$, and embed $G^*$ in the plane. The edge partition of $G^*$ into paths through opposite edges at each degree four vertex may be found in linear time by connected component analysis. By labeling each edge with the identity of its path, we may verify that this partition does not include cycles disjoint from $v_\infty$ and that no path crosses itself. We additionally check that $G$ has $\ell(\ell-1)/2$ vertices, where $\ell$ is the number of paths. Finally, by listing the pairs of paths passing through each vertex and bucket sorting this list, we may verify in linear time that no two paths cross more than once. If $G$ passes all of these checks, its decomposition into paths gives a valid pseudoline arrangement, which may be constructed by viewing the embedding of $G^*$ as being on a sphere, puncturing the sphere at point $v_\infty$, and homeomorphically mapping the punctured sphere to the plane. 
\end{proof}

\section{Small Grids}
\label{sec:grid}

\begin{figure}[b]
\centering\includegraphics[width=4in]{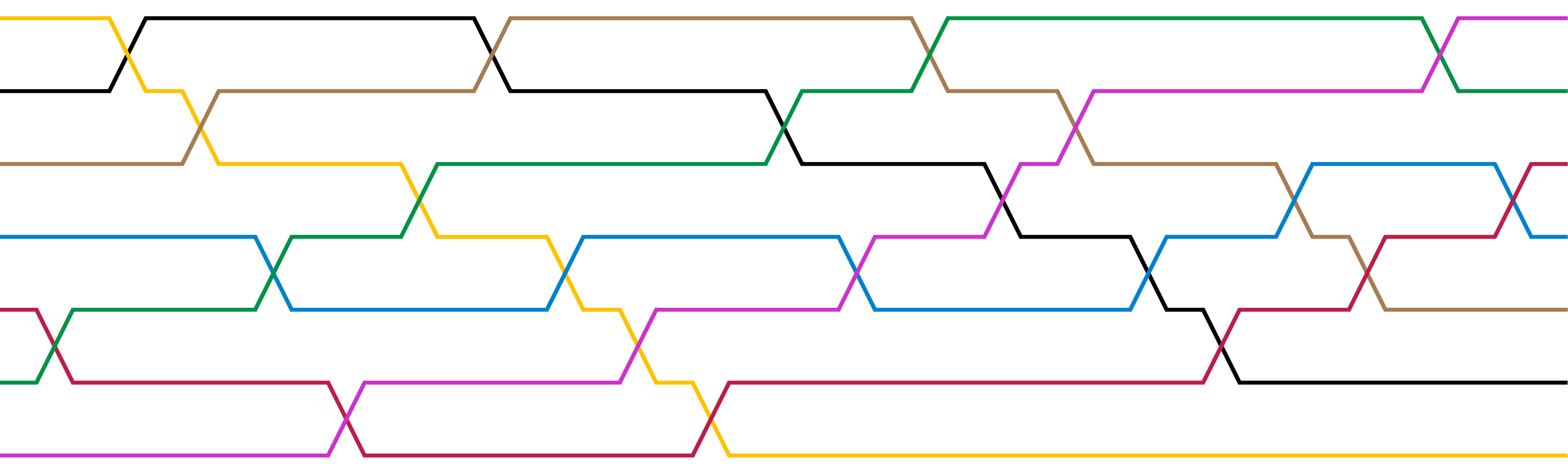}
\caption{A wiring diagram formed by a plane sweep of the arrangement from Figure~\ref{fig:planarity-solved}}
\label{fig:wiring-diagram}
\end{figure}

To describe our grid drawing algorithm for pseudoline arrangement graphs, we need to introduce the concept of a \emph{wiring diagram}. A wiring diagram is a particular kind of pseudoline arrangement, in which the $\ell$ pseudolines largely lie on $\ell$ horizontal lines (with coordinates $y=1$, $y=2$, $\dots$, $y=\ell$).  The pseudolines on two adjacent tracks may cross each other, swapping which track they lie on, near points with coordinates $x=1$, $x=2$, $\dots$, $x=\ell(\ell-1)/2$; each crossing is formed by removing two short segments of track and replacing them by two crossing line segments between the tracks. It is convenient to require different crossings to have different $x$ coordinates, following Goodman~\cite{Goo-DM-80}, although some later sources omit this requirement. Figure~\ref{fig:wiring-diagram} depicts an example. Wiring diagrams already provide reasonably nice grid drawings of arrangement graphs~\cite{MutSahPat-01}, but are unsuitable for our purposes, for two reasons: they draw the edges connecting pairs of adjacent crossings as polygonal chains with two bends, and for some arrangements,  even allowing crossings to share $x$-coordinates, drawing the wiring diagram of $\ell$ lines in a grid may require width $\Omega(\ell^2)$ (Figure~\ref{fig:wide-wiring}), much larger than our bounds. Instead, we will use these diagrams as a tool for constructing a different and more compact straight-line drawing.

\begin{figure}[t]
\centering\includegraphics[width=4in]{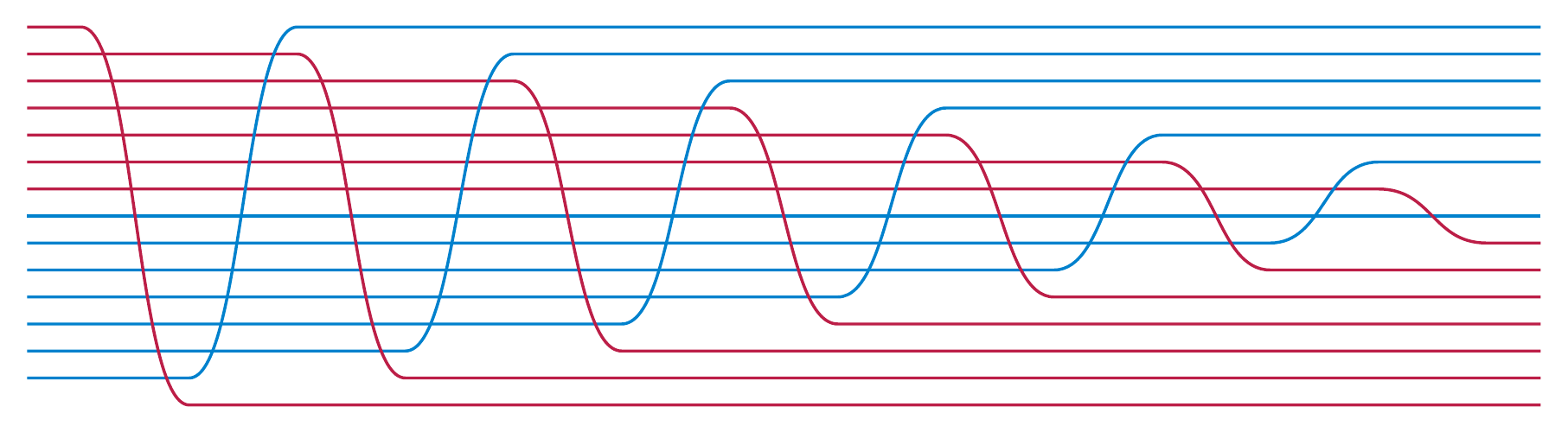}
\caption{Cocktail shaker sort corresponds to an arrangement of $\ell$ pseudolines for which drawing the wiring diagram in a grid requires width $\Omega(\ell^2)$}
\label{fig:wide-wiring}
\end{figure}

For an arrangement of non-vertical lines in general position, an equivalent wiring diagram may be constructed by a \emph{plane sweep} algorithm~\cite{BenOtt-IEEETC-79}, which simulates the left-to-right motion of a vertical line across the arrangement. At most points in the sweep, the intersection points of the arrangement lines with the sweep line maintain a fixed top-to-bottom order with each other, with their positions in this order reflected in the assignment of the corresponding pseudolines to tracks. When the sweep line crosses a vertex of the arrangement, two intersection points swap positions in the top-to-bottom order, corresponding to a crossing in the wiring diagram. The left-to-right order of crossings in the wiring diagram is thus exactly the sorted order of the crossing points of the arrangement, as sorted by their $x$ coordinates. The wiring diagram in Figure~\ref{fig:wiring-diagram} was constructed in this way from the approximate line arrangement depicted in Figure~\ref{fig:planarity-input}.

Every simple pseudoline arrangement, also, has an equivalent wiring diagram, that may be constructed in time linear in its number of crossings. The proof of this fact uses \emph{topological sweeping}, a variant of  plane sweeping originally developed to speed up sweeping of straight line arrangements by relaxing the strict left-to-right ordering of the crossing points~\cite{EdeGui-JCSS-89}, that can also be extended to apply to pseudoline arrangements~\cite{SnoHer-SCG-89}. The steps of the topological sweeping algorithm require only determining the relative ordering of crossings along each of the input pseudolines, something that may easily be determined from our path decomposition of a pseudoline arrangement graph by precomputing the position of each crossing on each of the two pseudolines it belongs to.

We define the $i$th \emph{level} $L_{\mathcal D}(i)$ in a wiring diagram $\mathcal D$ to be the set of crossings that occur between tracks $i$ and $i+1$. A crossing belongs to $L_{\mathcal D}(i)$ if and only if $i-1$ lines pass between it and the bottom face of the arrangement (the face below all of the tracks in the wiring diagram); therefore, once this bottom face is determined, the levels are fixed by this choice regardless of how the crossings are ordered to form a wiring diagram. If we define the size $|{\mathcal D}|$ of a diagram to be its number of pseudolines, and the level complexity $\klevel({\mathcal D})$ to be $\max_i |L_{\mathcal D}(i)|$, then it is a longstanding open problem in discrete geometry (a variant of the $k$-set problem) to determine the maximum level complexity of an arrangement of $\ell$ pseudolines,
$\maxklevel(\ell)=\max_{|{\mathcal D}|=\ell} \klevel({\mathcal D})$.
The known bounds on this quantity are
$\maxklevel(\ell)=O(\ell^{4/3})$~\cite{Dey-DCG-98,TamTok-Algo-03,ShaSmo-WADS-03},
and
$\maxklevel(\ell)=\Omega(\ell\,c^{\sqrt{\log\ell}})$ for some constant $c>1$~\cite{KlaPatPip-82,Tot-DCG-01},
where the last bound is $O(n^{1+\epsilon})$ for all constants $\epsilon>0$.

\begin{theorem}
\label{thm:draw}
Let $G$ be a pseudoline arrangement graph with $n$ vertices, determined by $\ell=\Theta(\sqrt n)$ pseudolines. Then in time $O(n)$ we may construct a planar straight-line drawing of $G$, in a grid of size $(\ell-1)\times\maxklevel(\ell)=O(n^{1/2})\times O(n^{2/3})$.
\end{theorem}

\begin{proof}
We find a decomposition of $G$ into pseudoline paths, by the algorithm of Lemma~\ref{lem:test-arrangement}, and use topological sweeping to convert this decomposition into a wiring diagram. We place each vertex $v$ of $G$ at the coordinates $(i,j)$, where $i$ is the position of $v$ within its level of the wiring diagram and $j$ is the number of tracks below its level of the wiring diagram.

With this layout, every edge of $G$ either connects consecutive vertices within the same level as each other, or it connects vertices on two consecutive levels. In the latter case, each edge between two consecutive levels corresponds to a horizontal segment of the wiring diagram that lies on the track between the two levels; the left-to-right ordering of these horizontal segments is the same as the left-to-right ordering of both the lower endpoints and the upper endpoints of these edges. Because of this consistent ordering of endpoints, no two edges between the same two consecutive levels can cross. There can also not be any crossings between edges that do not both lie in the same level or connect the same two consecutive levels. Therefore, the drawing we have constructed is planar. By construction, it has the dimensions given in the theorem.
\end{proof}

\begin{wrapfigure}[12]{r}{0.38\textwidth}
\vskip-6ex
\centering\includegraphics[scale=0.47]{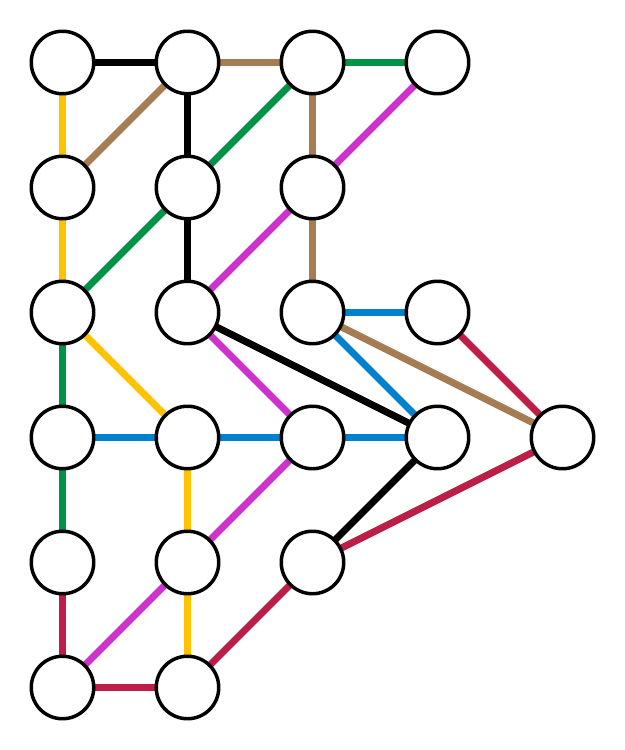}
\vskip-2ex
\caption{Output of the drawing algorithm of Theorem~\ref{thm:draw}, based on the wiring diagram of Figure~\ref{fig:wiring-diagram}}
\label{fig:algorithm-output}
\end{wrapfigure}
Figure~\ref{fig:algorithm-output} depicts the output of our algorithm, using the wiring diagram of Figure~\ref{fig:wiring-diagram}, for the graph of Figure~\ref{fig:planarity-input}.
Although the $6\times 5$ grid found by this algorithm is not quite as compact as the manually-found $5\times 5$ grid of Figure~\ref{fig:planarity-solved}, it is much smaller than standard grid drawings that do not take advantage of the arrangement structure of this graph. A more careful placement of vertices within each row would improve the angular resolution and edge length of the drawing but we have omitted this step in order to make the construction more clear.

\section{Universal Point Sets}

A \emph{universal point set} for the $n$-vertex graphs in a class $\mathcal C$ of graphs is a set $U_n$ of points in the plane such that every $n$-vertex graph in $C$ can be drawn with its vertices in $U_n$ and with its edges drawn as non-crossing straight line segments~\cite{ChrKar-SN-89}.
Grids of $O(n)\times O(n)$ points form universal sets of quadratic size for the planar graphs~\cite{FraPacPol-STOC-88,Sch-SODA-90}, and despite very recent improvements to the constant factor in this quadratic size bound~\cite{BanCheDev-13} this remains the best known upper bound. A rectangular grid that is universal must have $\Omega(n^2)$ points~\cite{DolLeiTri-ACR-84,Val-TC-81}; the best known lower bounds for universal point sets that are not required to be grids are only linear~\cite{ChrKar-SN-89}.

Subquadratic bounds are known on the size of universal point sets for subclasses of the planar graphs including the outerplanar graphs~\cite{GriMohPol-AMM-1991}, simply-nested planar graphs~\cite{AngDibKau-GD-2012,BanCheDev-13}, planar 3-trees~\cite{FulTot-WADS-2013}, and graphs of bounded pathwidth~\cite{BanCheDev-13}; however, these results do not apply to arrangement graphs. The grid drawing technique of Theorem~\ref{thm:draw} immediately provides a universal point set for arrangement graphs of size $O(n^{7/6})$; in this section we significantly improve this bound, while only increasing the area of our drawings by a constant factor.

Following Bannister et al.~\cite{BanCheDev-13}, define a sequence of positive integers $\xi_i$ for $i=1,2,3,\dots$ by the equation $\xi_i=i\oplus(i-1)$ where $\oplus$ denotes the bitwise binary exclusive or operation. The sequence of these values begins
$$1,3,1,7,1,3,1,15,1,3,1,7,1,3,1,\dots.$$

\begin{lemma}[Bannister et al.~\cite{BanCheDev-13}]
\label{lem:sawtooth}
Let the finite sequence $\alpha_1,\alpha_2,\dots \alpha_k$ have sum $s$. Then there is a subsequence $\beta_1,\beta_2,\dots \beta_k$ of the first $s$ terms of $\xi$ such that, for all $i$, $\alpha_i\le \beta_i$. The sum of the first $s$ terms of $\xi$ is between $s\log_2 s-2s$ and $s\log_2 s+s$.
\end{lemma}

Recall that the grid drawing technique of Theorem~\ref{thm:draw} produces a drawing in which the vertices are organized into $\ell-1$ rows of at most $\maxklevel(\ell)=O(\ell^{4/3})$ vertices per row, where $\ell=O(\sqrt n)$ is the number of lines in the underlying $n$-vertex arrangement. In this drawing, suppose that there are $n_i$ vertices on the $i$th row of the drawing, and define a sequence $\alpha_i=\lceil n_i/\ell\rceil$.

\begin{lemma}
$\sum\alpha_i\le 3(\ell-1)/2$.
\end{lemma}

\begin{proof}
We may partition the $n_i$ vertices in the $i$th row $n_i$ into $\lfloor n_i/\ell\rfloor$ groups of exactly $\ell$ vertices, together with at most one smaller group; then $\alpha_i$ is the number of groups. The contribution to $\sum\alpha_i$ from the groups of exactly $\ell$ vertices is at most $n/\ell=(\ell-1)/2$. There is at most one smaller group per row so the contribution from the smaller groups is at most $\ell-1$. Thus the total value of the sum is at most $3(\ell-1)/2$.
\end{proof}

\begin{theorem}
There is a universal point set of $O(n\log n)$ points for the $n$-vertex arrangement graphs,
forming a subset of a grid of dimensions $O(\ell)\times \maxklevel(\ell)$.
\end{theorem}

\begin{proof}
Let $s=3(\ell-1)/2$.
We form our universal point set as a subset of an $s\times\maxklevel(\ell)$ grid; the area of the grid from which the points are drawn is exactly 3/2 times the area of the $(\ell-1)$-row grid drawing technique of Theorem~\ref{thm:draw}. In the $i$th row of this grid, we include in our universal point set $\min(\ell\xi_i,\maxklevel(\ell))$ of the grid vertices in that row. It does not matter for our construction exactly which points of the row are chosen to make this number of points.

By Lemma~\ref{lem:sawtooth}, there is a subsequence $\beta_i$ of the first $s$ rows of sequence $\xi$, such that the $\beta$ is termwise greater than or equal to $\alpha$. This subsequence corresponds to a subsequence
$(r_1,r_2,\dots r_{\ell-1})$ of the rows of our universal point set, such that row $r_i$ has at least $\min(\ell\beta_i,\maxklevel(ell))\ge n_i$ points in it.
Mapping the $i$th row of the drawing of Theorem~\ref{thm:draw} to row $r_i$ of this point set will not create any crossings, because the mapping is monotonic within each row and because all edges of the drawing connect pairs of vertices that are either in the same row or in consecutive rows.

The number of points in the point set is $O(\ell s\log s)$ where $s=\sum\alpha_i=O(\ell)$.
Therefore, this number of points is $O(\ell^2\log\ell)=O(n\log n)$.
\end{proof}

\section{Greedy embedding algorithm}

The algorithm of Lemma~\ref{lem:test-arrangement} uses as a subroutine a linear-time planarity testing algorithm. Although such algorithms may be efficiently implemented on computers, they are not really suitable for hand solution of Planarity puzzles.
Instead, it is more effective in practice to build up a planar embedding one face at a time, by repeatedly finding a short cycle in the input graph and attaching it to the previously constructed partial embedding. Here ``short'' means as short as can be found; it is not possible to limit attention to cycles of length three, four, or any fixed bound. For instance in Figure~\ref{fig:unstretchable} the central triangle is separated from the rest of the graph by faces with five sides, and by modifying this example it is possible to separate part of an arrangement graph from the rest of the graph by faces with arbitrarily many sides.
Thus, this hand-solution heuristic may be formalized by the following steps.
\begin{enumerate}
\item Choose an arbitrary starting vertex $v$.
\item Find a cycle $C_1$ of minimum possible length containing $v$.
\item Embed $C_1$ as a simple cycle in the plane.
\item While some of the edges of the input graph have not yet been embedded:
\begin{enumerate}
\item Let $C_i$ be the cycle bounding the current partial embedding. Define an \emph{attachment vertex} of $C_i$ to be a vertex that is incident with edges not already part of the current embedding.
\item Choose two attachment vertices $u$ and $v$, and a path $P_i$ in $C_i$ from $u$ to $v$, such that there are no attachment vertices interior to $P_i$.
\item Find a shortest path $S_i$ from $u$ to $v$, using only edges that are not already part of the current partial embedding.
\item If necessary, adjust the positions of the embedded vertices (without changing the combinatorial structure of the embedding) so that $S_i$ may be drawn with straight line edges.
\item Add $S_i$ to the embedding, outside $C_i$, so that the new face between $P_i$ and $S_i$ does not contain $C_i$. After this change, the new bounding cycle $C_{i+1}$ of the partial embedding is formed from $C_i$ by replacing $P_i$ by $S_i$.
\end{enumerate}
\end{enumerate}

When it is successful, this algorithm decomposes the input graph into the cycle $C_1$ and a sequence of edge-disjoint paths $S_1$, $S_2$, etc. Such a decomposition is known as an \emph{open ear decomposition}~\cite{Khu-SN-89}.

This greedy ear decomposition algorithm does not always work for arbitrary planar graphs: even the initial cycle that is found by the algorithm may not be a face of an embedding of the given graph, causing the algorithm to make incorrect assumptions about the structure of the embedding. 
However (ignoring the possible difficulty of performing step d) the algorithm does always correctly embed the arrangement graphs used by Planarity. These graphs may have multiple embeddings; to distinguish among them, define the \emph{canonical embedding} of an arrangement graph to be the one given by the arrangement from which it was constructed. By Lemma~\ref{lem:test-arrangement}, the canonical embedding is unique.
 As we prove below, the cycles of an arrangement graph that the algorithm assumes to be faces really are faces of the canonical embedding.

\begin{figure}[t]
\centering\includegraphics[scale=0.4]{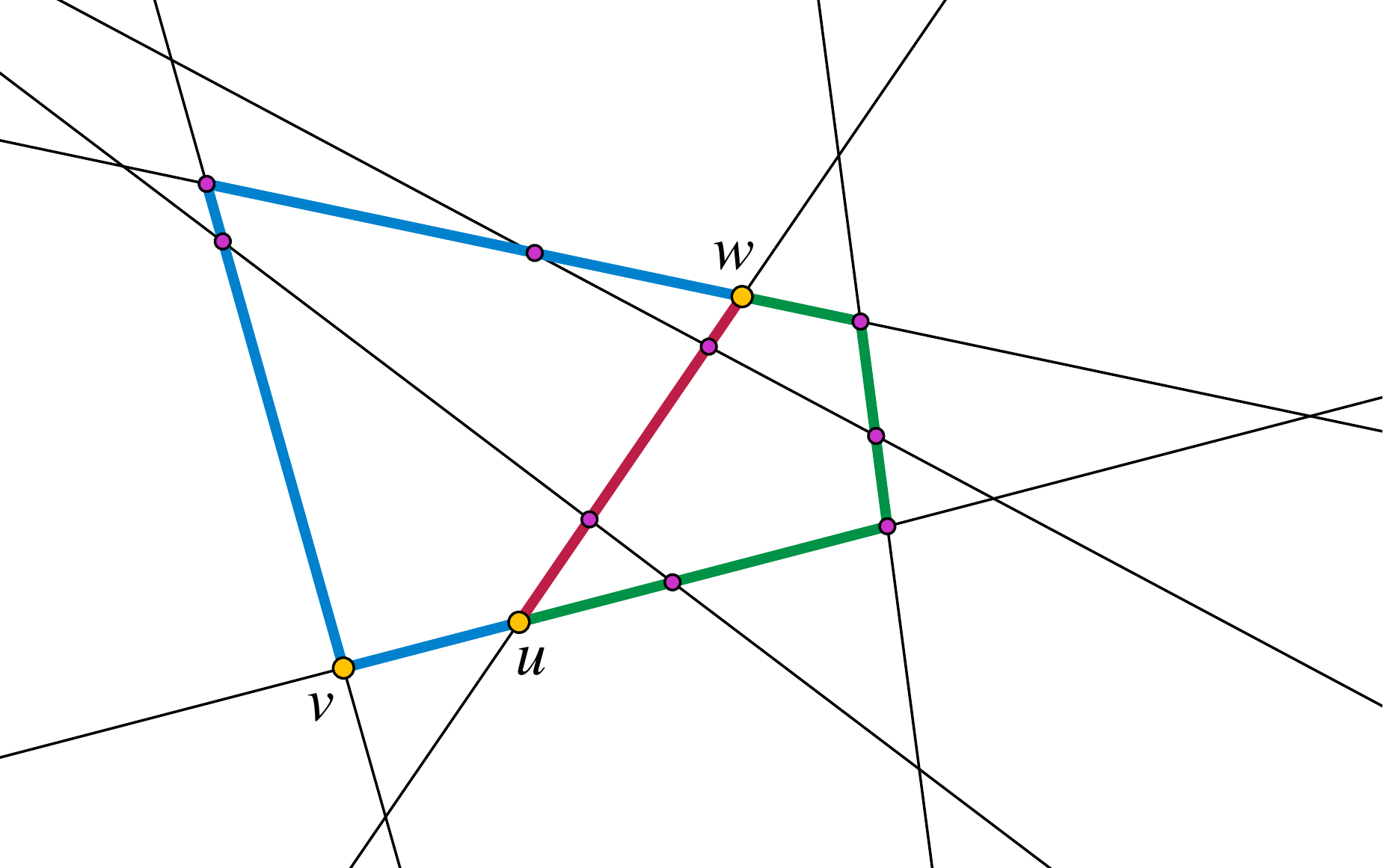}
\caption{Illustration for the proof of Lemma~\ref{lem:first-step}. Every non-facial cycle $C$ through vertex $v$ (blue and green edges) is crossed by at least one line $\ell=uw$ (red edges), forming a theta-graph. All the vertices on the red path of the theta are matched by an equal number of vertices on each of the other two paths, caused by crossings with the same lines, and the other two paths have additional vertices at their bends, so the red-blue cycle is shorter than the blue-green cycle.}
\label{fig:theta}
\end{figure}

\begin{lemma}
\label{lem:first-step}
Let $v$ be an arbitrary vertex of arrangement graph $G$, and $C$ be a shortest cycle containing $v$. Then $C$ is a face of the canonical embedding of~$G$.
\end{lemma}

\begin{proof}
Let $C$ be an arbitrary simple cycle through $v$. Then if $C$ is not a face of the arrangement forming~$G$, there is a line $\ell$ that crosses it; let $u$ and $w$ be two vertices on the boundary of $C$ connected through the interior of $C$ by $\ell$ (Figure~\ref{fig:theta}). Then $C$ together with the path along $\ell$ from $u$ to $w$ form a theta-graph, a graph with two degree three vertices ($u$ and $w$) connected by three paths.  Every vertex of $\ell$ between $u$ and $w$ is caused by a crossing of $\ell$ with another line that also must cross the other two paths of the theta-graph; in addition, each of these two paths must bend at least once at a vertex that does not correspond to a line that crosses $\ell$.
Therefore, the path through $\ell$ is strictly shorter than the other two paths in the theta-graph. Replacing one of the two paths of $C$ from $u$ to $w$ by the path through $\ell$ produces a shorter cycle that still contains $v$. Since an arbitrary cycle $C$ that is not a face can be replaced by a shorter cycle through $v$, it follows that every shortest cycle through $v$ is a face.
\end{proof}

\begin{lemma}
\label{lem:ear}
Let $D$ be a drawing of a subset of the faces of the canonical embedding of an arrangement graph $G$ whose union is a topological disk, let $u$ and $v$ be two attachment vertices on the boundary of $D$ with no attachment vertices interior to the boundary path $P$ from $u$ to $v$, and let $S$ be a shortest path from $u$ to $v$ using only edges not already part of $D$. Then the cycle formed by the union of $P$ and $S$ is a face of the canonical embedding of~$G$.
\end{lemma}

\begin{proof}
Assume for a contradiction that $P\cup S$ is not a face; then as in the proof of Lemma~\ref{lem:first-step}, this cycle must be crossed by a line $\ell$, a path $L$ of which forms a theta-graph together with $P\cup S$. Additionally, because $P$ is assumed to be part of a drawing of a subset of the faces of $G$, it cannot be crossed by $\ell$, for any crossing would cause it to have an attachment vertex between $u$ and $v$. Therefore, the two degree-three vertices of the theta-graph both belong to $S$. By the same reasoning as in the proof of Lemma~\ref{lem:first-step}, $L$ must be strictly shorter than the other two paths of the theta-graph, so replacing the path that is entirely within $S$ by $L$ would produce a shorter path from $u$ to $v$, contradicting the construction of $S$ as a shortest path. This contradiction shows that $P\cup S$ must be a face, as the lemma states.
\end{proof}

\begin{theorem}
When the greedy ear decomposition embedding algorithm described above is applied to an arrangement graph~$G$, it correctly constructs the canonical embedding of~$G$.
\end{theorem}

\begin{proof}
We use induction on the number of steps of the algorithm, with an induction hypothesis that after each step the partial embedding found so far consists of faces of the canonical embedding whose union is a disk. Lemma~\ref{lem:first-step} shows as a base case that the induction hypothesis is true after the first step. In each subsequent step, the ability to find two attachment vertices follows from the fact that arrangement graphs are 2-vertex-connected, which in turn follows from the fact that they can be augmented by a single vertex to be 3-vertex-connected~\cite{BosEveWis-IJCGA-03}. Lemma~\ref{lem:ear} shows that, if the induction hypothesis is true after $i$ steps then it remains true after $i+1$ steps. 
\end{proof}

\section{Conclusions}

We have found a grid drawing algorithm for pseudoline arrangement graphs that uses area within a small factor of linear, much smaller than the known quadratic grid area lower bounds for arbitrary planar graphs. We have also shown that these graphs have near-linear universal point sets within a constant factor of the same area, and that a simple greedy embedding heuristic suitable for hand solution of Planarity puzzles is guaranteed to find a correct embedding.

\begin{figure}[t]
\centering\includegraphics[scale=0.5]{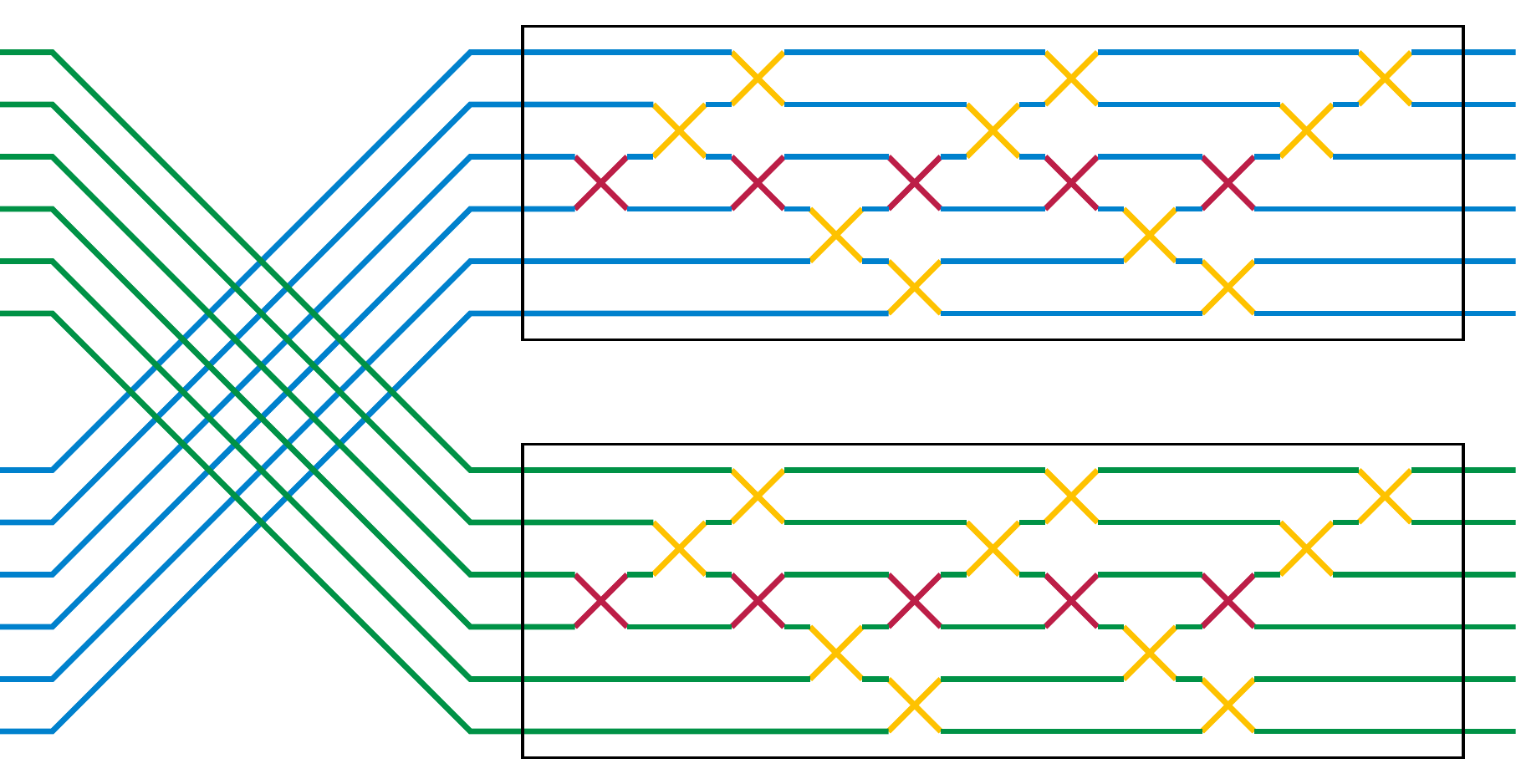}
\caption{Two stacked arrangements of $\ell/2$ pseudolines, each with high level complexity, cause our algorithm to create wide drawings no matter how it chooses a wiring diagram.}
\label{fig:stacked-worst-case}
\vspace{-2ex}
\end{figure}

The precise area used by our grid drawing algorithm depends on the worst-case behavior of the function $\klevel({\mathcal D})$ counting the number of crossings in a $k$-level of an arrangement; closing the gap between the upper and lower bounds for this function remains an important and difficult open problem in combinatorial geometry. However, closing this gap is not the only possible method for improving our drawing algorithm.

A tempting avenue for improvement is to observe that a single pseudoline arrangement may be represented by many different wiring diagrams; therefore, we can select the wiring diagram $\mathcal{D}$ that represent the same pseudoline arrangement and that minimizes $\klevel({\mathcal{D}})$. However, this would not improve our worst case width by more than a constant factor. For, if the input forms a pseudoline arrangement constructed by stacking two arrangements of $\ell/2$ lines with maximal $k$-level complexity, one above the other (Figure~\ref{fig:stacked-worst-case}), then one of these two instances will survive intact in any wiring diagram for the arrangement, forcing our algorithm to produce a drawing with width at least $\maxklevel(\ell/2)$. Further improvements in our algorithm will likely come by finding an alternative layout that avoids the complexity of $k$-levels, by proving that $k$-levels are small in the average case if not the worst case, or by reducing the known combinatorial bounds on $k$-levels.

{\raggedright
\bibliographystyle{abuser}
\bibliography{planarity}}

\end{document}